\documentclass[11pt]{article}
\usepackage{amssymb,amsmath}
\usepackage{epsfig}
\usepackage{palatino}
\usepackage{graphicx}
\usepackage{textcomp}
\usepackage{hyperref}

\newtheorem{theorem}{Theorem}

\newtheorem{definition}{Definition}

\newtheorem{lemma}{Lemma}
\newtheorem{conjecture}{Conjecture}

\newtheorem{corollary}{Corollary}
\newtheorem{fact}{Fact}

 \newcommand{\qedsymb}{\hfill{\rule{2mm}{2mm}}}  
 \newenvironment{proof}[1][]{\begin{trivlist}  
 \item[\hspace{\labelsep}{\bf\noindent Proof#1:\/}] 
 }{\qedsymb\end{trivlist}}

\newcommand{\be}{\begin{eqnarray}}
\newcommand{\ee}{\end{eqnarray}}

\newcommand\ket[1]{{ |{#1} \rangle }}

\def\QMA{{\sf{QMA}}}
\def\CLH{{\sf{CLH}}}
\def\PCP{{\sf{PCP}}}
\def\qPCP{{\sf{qPCP}}}
\def\NLTS{{\sf{NLTS}}}
\def\NP{{\sf{NP}}}

\def\P{{\sf{P}}}
\def\LH{{\sf{LH}}}
\def\CSP{{\sf{CSP}}}

\newcommand{\ignore}[1]{}

\newcommand{\eps}{\varepsilon}
\renewcommand{\epsilon}{\varepsilon}

\title{The commuting local Hamiltonian on locally-expanding graphs is in $\NP$.}
\begin{document}

\author{Dorit Aharonov\thanks{School of Computer Science and
Engineering, The Hebrew University,
Jerusalem, Israel}
 \and Lior Eldar\thanks{School of Computer Science and
Engineering, The Hebrew University,
Jerusalem, Israel.}}

\date{\today}

\maketitle


\begin{abstract}
The local Hamiltonian problem is famously 
complete for the class $\QMA$, 
the quantum analogue of $\NP$ \cite{Kit}. The complexity of its 
semi-classical version, in which 
the terms of the Hamiltonian are required to  
commute (the $\CLH$ problem), 
has attracted considerable 
attention recently \cite{Aha,Bra,Has,Sch, Has2}
due to its intriguing nature, as well as in 
relation to growing interest in the $\qPCP$ conjecture \cite{Aha2,AAV}.  
We show here that if the underlying bipartite interaction
graph of the $\CLH$ instance is a good locally-expanding graph, namely, 
the expansion of any constant-size set  
is $\epsilon$-close to optimal, then
approximating its ground energy to within additive factor $O(\epsilon)$ lies in $\NP$.
The proof holds for $k$-local Hamiltonians for
 any constant $k$ and any constant 
dimensionality of particles 
$d$. 
We also show that the approximation problem of $\CLH$ on such good local 
expanders is $\NP$-{\it hard}. 
This implies that too good local expansion of the interaction graph 
constitutes an obstacle 
against {\it quantum} hardness of the approximation 
problem, though it retains its classical hardness.
The result 
highlights new difficulties in trying to mimic classical 
proofs  (in particular Dinur's $\PCP$ proof \cite{Din}) 
in an attempt to prove the quantum $\PCP$ conjecture.  
A related result 
was discovered recently independently  
by Brand{\~a}o and Harrow \cite{HB}, for $2$-local general Hamiltonians, 
bounding the quantum hardness of the approximation problem on good 
expanders, though no $\NP$-hardness is known in that case. 
\end{abstract}

\section{Introduction} 
Quantum Hamiltonian complexity (QHC) 
has blossomed into an incredibly active field of
research over the past few years (see \cite{Osb,AAV}
for an introduction). 
A key player in QHC is the
local Hamiltonian problem ($\LH$) \cite{Kit},
whose study had led to various unexpected discoveries, 
ranging from the universality of adiabatic computation \cite{Aha3}, 
to the $\QMA$ hardness of one dimensional systems \cite{Aha4}, 
to quantum gadgets \cite{KKR,TO,Terhalk}, 
to computational complexity 
lower bounds on notorious open problems in physics \cite{
SV},
to the hardness of translationally invariant systems \cite{Got}, and more.  

In this note we focus on a special case of the $\LH$ problem, 
called the  
commuting local Hamiltonian problem ($\CLH$), which has attracted 
considerable attention in the past few years 
(e.g.,\cite{Aha,Bra,Has,Sch, Has2}). Before we explain 
our result, let us define the problem more precisely. 
A $k$-local Hamiltonian is a Hermitian matrix $H = \sum_i H_i$ 
operating on the Hilbert space 
of $n$ $d$-dimensional particles ({\it qudits}), where each term $H_i$ 
(sometimes called constraint) acts non-trivially 
on at most 
$k$ particles. 
Kitaev defined in $1999$ 
the local Hamiltonian ($\LH$) problem \cite{Kit}, where 
one is given 
a local Hamiltonian on $n$ 
qudits, and two real numbers $a,b$, with $b-a\ge 1/poly(n)$,  
and is asked 
whether the lowest eigenvalue ({\it ground energy}) 
is at most $a$ or at least $b$. 
Kitaev 
showed \cite{Kit} that this problem
is complete for the class $\QMA$,  
the quantum analogue of $\NP$. 
This provided a quantum counterpart of 
the celebrated Cook-Levin theorem \cite{P}.
To restrict the discussion to {\it commuting} local Hamiltonians, 
we require in addition that 
every two terms $H_i, H_j$ commute. W.L.O.G we can assume in this case that 
$H_i$'s are projections. 
We denote the family of such Hamiltonians acting on $d$ dimensional 
particles and where each term is $k$-local, 
by $\CLH(k,d)$; 
We can now consider the $\CLH$ 
problem (introduced by Bravyi and Vyalyi in $2003$ \cite{Bra}), 
which is the analogue of the $\LH$ problem for the commuting case. 
Note that since we are concerned with $H_i$'s which are commuting projections, 
the eigenvalues of $H$ are natural numbers. 
The problem is now to distinguish between the ground energy being   
$a=0$ or at least $b=1$.   

The commutation restriction might seem at first sight to devoid the 
$\LH$ problem of its quantum nature, and place it in $\NP$, 
since all local terms can be diagonalized simultaneously. 
Moreover, one usually attributes the interesting features of quantum 
mechanics to its non-commutative nature, 
cf. the Heisenberg uncertainty principle. 
This intuition is importantly wrong. 
The main source for the interest in $\CLH$s lies 
in the beautiful and striking fact that
the groundspaces of $\CLH$s are capable of possessing intriguing
multi-particle entanglement phenomena, such as topological order \cite{Kit2}, 
and quantum error correction \cite{Got};  
the most well known example is Kitaev's toric code 
\cite{Kit2}. 
Due to this intriguing phenomenon, 
the study of $\CLH$s has  
become highly influential in the physics as well as the 
quantum complexity communities, see e.g. the study of
Levin-Wen models \cite{LevinWen}, 
quantum double models \cite{Kit2}, as well as 
that of stabilizer error correcting 
codes \cite{Got}, which are special cases of $\CLH$'s. 

In their $2003$ paper, 
Bravyi and Vyalyi \cite{Bra} 
proved a central result regarding the $\CLH$ problem: 
the $2$-local case (namely, the case of $k=2$) 
lies in $\NP$. 
The proof uses a
clever application of 
the theory of representations of $C^*$-algebras, with which Bravyi and 
Vyalyi showed 
that for such systems, there is always a ground state which can be generated by
a constant 
depth quantum circuit; 
in particular, the description 
of the constant depth quantum circuit can be given as a classical witness, 
from which the energy of the state can be computed efficiently by a classical computer
(implying the 
containment in $\NP$). 
The $\CLH$ question is thus interesting (from the 
quantum complexity point of view) only with locality $k>2$. 
The Bravyi-Vyalyi result was extended later in various special cases
\cite{Aha, Has, Sch, Has2}, but  
the complexity of the general  
$\CLH$ problem remains open, and as far as we currently know it may 
lie anywhere between $\NP$ and the full power of $\QMA$.
\footnote{To be more precise, $\QMA_1$ - 
which is defined like $\QMA$ except in case of YES, 
there exists a state which is accepted with probability exactly $1$. 
Here we will not distinguish between $\QMA$ and $\QMA_1$
since we are not using those terms technically, but 
see \cite{Aar} and \cite{Bra06}.} 

In this paper, we study the hardness of {\it approximation} versions 
of the $\CLH$ 
problem. Our motivation is two-fold: one, we view this 
as an important stepping stone towards the 
clarification of the complexity of the exact $\CLH$ problem, 
which has resisted resolution so far,  despite the effort 
described above.   
Two, we view it as an important special case of a major conjecture in QHC, 
namely, the quantum $\PCP$ ($\qPCP$) conjecture, to be  
discussed below. 

Our main result states that the 
approximation problem of $\CLH$ 
lies within $\NP$, whenever the underlying bi-partite 
interaction graph of the $\CLH$ instance 
is a good enough {\it locally-expanding graph} (for definition see below).
This means that good local expansion of the underlying 
interaction graph does 
not help but in fact {\it disturbs} in 
generating ``hard'' quantum instances, in sharp contrast to 
the classical analogue in which the hard instances are excellent 
locally-expanding graphs (see \cite{Din} and Theorem \ref{thm:NPhard} 
below). 
The result implies
severe constraints on the structure of possible $\qPCP$s for $\CLH$s, 
as the output of those cannot be good local expanders.  
Before we state the result and its implications more formally, 
we provide some further background and context. 

\subsection{Relation to the Quantum $\PCP$ conjecture} 
The classical $\PCP$ theorem \cite{Arora, ALMSS}
is arguably the most important 
discovery in classical theoretical computer science over the past quarter 
century; 
To state it, recall 
the Constraint Satisfaction Problem ($\CSP$):  
Given is a collection of $k$-local constraints on $n$ 
Boolean or $d$-state variables. The problem is to decide 
whether the instance is satisfiable or not; this problem is $\NP$ complete. 
The $\PCP$ theorem states that there exists an efficient mapping 
(known as a $\PCP$ reduction) 
which takes satisfiable instances to satisfiable 
instances and non-satisfiable ones to 
non-satisfiable ones, such that the $\NP$ witness for the new instances  
can be tested, 
with constant probability of success, 
by reading only a constant number of randomly chosen locations 
in the witness (!).  
Stated differently, the $\PCP$ theorem says that 
it is $\NP$ hard to decide whether 
a given instance is satisfiable, or at most 
some constant fraction, say $90\%$, of its constraints can be satisfied. 
One may ask \cite{AN,Aar}
does an analogous statement hold in the quantum setting as well? 
A quantum $\PCP$ conjecture \cite{Aha2}     
can be phrased as follows:
\begin{conjecture} {\bf Quantum $\PCP$ ($\qPCP$)}
There exists $c>0$, such that it is $\QMA$-hard 
to decide whether a given $k$-local Hamiltonian 
$H = \sum_{i=1}^m H_i$, with $\left\|H_i \right\|\leq 1$ whose terms are 
all positive semi-definite, 
is satisfiable (namely, has ground value $0$) or 
its minimal eigenvalue is 
at least $c \cdot m$.
\end{conjecture}

This conjecture is one of the major open problems in QHC today 
(see the recent survey \cite{AAV}). 
It is important for several reasons.  
First of all, the question is related to the question of 
whether quantum entanglement can be made robust at room temperature. 
Secondly, the question of whether or not $\qPCP$s exist 
is tightly related to fundamental questions regarding local and global
multi-particle entanglement. The resolution of the conjecture, in
either direction, would 
constitute a major advance in QHC and in our understand of quantum 
matter. However, despite considerable effort, so far 
all attempts to extend the classical proofs of the 
$\PCP$ theorem to the quantum settings, or to dispute the conjecture, 
have encountered 
severe obstacles \cite{Aha2, Arad, Has, Has2, HB}. 
We currently do not even have a good 
guess whether the conjecture should hold or not. 
 
Just like in the classical case, 
the $\qPCP$ conjecture is equivalent to the 
following statement: 
the problem of    
approximating the ground energy of local Hamiltonians 
to within a constant fraction is $\QMA$ hard. 
This means that the clarification of the complexity of 
approximating the local Hamiltonian problem is of major interest.  
Given the difficulty of this question, it is natural  
to study hardness of approximation of restricted classes 
Hamiltonians;  we focus here on   
$\CLH$s, which are a particularly attractive special case,  
since they are both much easier to analyze than the general case, 
and in addition, they possess a simplifying   
structure as introduced in \cite{Bra}, while still giving rise 
to extremely complex phenomenon such as topological order, 
and used as a substrate for almost all known quantum codes.

\subsection{Computational Hardness and the geometry of the interaction graph}
The geometry of the underlying interaction graph is well 
known to play a crucial role for the complexity of constraint satisfaction
problem. For a start, when a $\CSP$ 
is defined between variables set a one dimensional lattice, with constraints 
between nearest neighbors only, the problem lies in
$\P$ (using dynamic programming), 
whereas for $2$ or more dimensions, the problem is $\NP$-hard.   
When one is interested in approximations (to within any constant
fraction of the number of constraints) the situation changes:  
approximating $\CSP$s, as well as $\LH$s, to within any 
constant, lies in $\P$ when the interaction graph is a lattice 
of {\it any} constant dimension. 
To see this, partition the lattice to constant size cubes,  
and remove the constraints that connect different cubes; 
we arrive at a collection 
of disconnected constant-size 
subproblems, all of which can be optimized 
together in $\P$. The number of constraints violated by this solution 
is at most the number 
of removed constraints, which can be made to be an arbitrarily small 
fraction. 
This argument works similarly both in the quantum and classical cases, 
and implies that to make the approximation problem hard, one has to 
consider $\CSP$s, $\CLH$s or $\LH$s 
whose underlying graphs have certain expansion properties.

In light of the above, and motivated by the important role that  
expanders play in the classical theory of hardness of approximation 
and $\PCP$s, we turn to the study of approximating $\CLH$ systems 
on expanders.  
We need our definition of expansion
to work for $k$-local Hamiltonians, for $k>2$, since 
$2$-local $\CLH$s are already 
known to be in $\NP$ even without approximation \cite{Bra}. 
One might consider working with the naturally induced hypergraphs, 
but there is no consensus on the right
definition of expansion for hypergraphs in the literature. 
Moreover, whereas one is perhaps used to intuitively think of standard graph 
expansion as the 
important property for hardness in classical complexity,  
perhaps due to the strong usage of expansion in Dinur's $\PCP$ proof 
\cite{Din}, it is in fact an open question to relate expansion, and 
in particular {\it small-set} expansion, (namely expansion of 
sets of constant fractional size) to 
computational hardness \cite{Ragh, Ragh2}.

Here we consider, for a given $\CLH$, the bi-partite graph in which 
variables are on the left side and constraints are on the right, 
and a variable is connected to the constraints that it appears in.
For this bi-partite graph, we consider a very {\it weak} notion of 
expansion, which we call {\it local expansion}: 
we say that the bi-partite graph 
induced by a $k$-local $\CLH$ instance is $\epsilon$-locally 
expanding, if for any set $S\subseteq L$ of at most $k$ particles, 
the number of local 
terms incident on these particles is at least $D_L |S| (1-\epsilon)$, where 
$D_L$ is the left degree of the graph (here we assume all degrees are equal, 
see Definition \ref{def:expbi} for the general case).   
Observe that 
$D_L |S|$ is the maximal possible number of constraints acting on those 
particles, so $\epsilon$ can be viewed a ``correction'' to this number, 
or the {\it expansion error}.  
We note that this definition of local expansion 
cares only about the expansion of  
constant-size sets, whereas in complexity theory, 
one is often interested in {\it small-set} bi-partite expanders, in which 
all sets of size up to some
{\it constant fraction} of $|L|$ are required to expand \cite{Din3,Ragh}; 
our requirements on bi-partite expanders are thus 
significantly weaker \footnote{
Of course, our results regarding $\CLH$s are stronger if they 
hold for graphs which satisfy weaker requirements}.
It turns out that the notion of 
local expansion can indeed be related, in the classical world, 
to computational hardness; 
as we show in Theorem \ref{thm:NPhard} sufficiently   
good local expanders are computationally hard for classical computers.  
As we will see (in Theorem \ref{thm:approxBPinNP}) 
this is not the case in the quantum world\footnote{Assuming standard computational complexity assumptions, specifically, $\NP \subsetneq \QMA_1$.}. 

\subsection{Main Result}
We can now state our main result, 
Theorem \ref{thm:approxBPinNP}, in which we show that 
the problem of approximating $\CLH$s on locally-expanding graphs lies inside 
$\NP$, 
and thus cannot be $\QMA$ hard (under standard complexity assumptions): 

\begin{theorem}\label{thm:approxBPinNP}
(the approximation of 
$\CLH$ on good locally-expanding bi-partite graphs is in $\NP$) 
Let $\gamma(\eps) = 2kd \eps$.
Let $H= \sum_{i=1}^m H_i$, be an instance of $\CLH(k,d)$ 
for constants $k,d>0$, with bi-partite 
interaction graph
which is $\epsilon$-locally-expanding , for $\epsilon<\frac{1}{2}$. 
Then the $\gamma(\epsilon)$-approximation problem of $\CLH(k,d)$ 
on such $\epsilon$ locally-expanding bi-partite graph graphs is in $\NP$. 
Moreover, for any eigenvalue $\lambda$ of $H$ there exists 
a state $\ket{\psi}$
such that 
$\left| \langle \psi| H |\psi \rangle - \lambda \right| \leq \gamma(\eps) \cdot m$, such that 
$\ket{\psi}$ can be generated by a constant depth quantum circuit.  
\end{theorem}

We rule out the possibility that 
$\gamma(\epsilon)$ is large enough so that by removing a $\gamma(\epsilon)$
fraction of the terms in $H$ 
we can simply trivialize the problem and place it in $\P$
(e.g., by disconnecting 
the graph into small components, as is done in the approximation of $\CLH$s 
on lattices). We do this by showing that the problem remains 
$\NP$-hard even with this 
approximation factor. 
  
\begin{theorem}\label{thm:NPhard}
Fix $k\ge 3,d\ge 2$. 
Let $\epsilon$ be such that $\gamma(\epsilon)=2kd\epsilon\le C_0$
where $C_0<1/3$ is a universal constant.  
Then it is $\NP$-hard 
to approximate $\CLH(k,d)$ 
whose bi-partite interaction graph is
$\eps$ locally-expanding, to within a factor 
$\gamma(\eps)=2kd\eps$. 
\end{theorem}
$C_0$ in the above statement is some 
universal constant which appears in the $\PCP$ proof of 
Dinur \cite{Din}; we delay its exact definition to Subsection \ref{sec:NPhard} 
but it is not important. 
Theorem \ref{thm:approxBPinNP} implies that 
if indeed one could construct a quantum $\PCP$ whose output is a 
$\CLH$, then in contrast
to classical $\PCP$ (see the proof of Theorem \ref{thm:NPhard}), 
the output cannot have arbitrarily high local expansion.

\subsection{Related work}
\label{subsec:compare}
The question of approximating the ground energy of local Hamiltonians
was broadly studied in the physics literature for specific 
Hamiltonians of 
particular interest in physics.  
Not much is known regarding the hardness of approximation of the 
general case. Kempe and Gharibian \cite{GK1} 
derived a non-trivial result for a related task: 
given a $k$-local Hamiltonian on $d$ dimensional qudits,
and we are asked to find the state with {\it maximal} energy $OPT$. 
They showed that there always exists 
a product state with energy at least $OPT/d^{k-1}$, 
and so the problem of approximating
MAX-$k$-local Hamiltonian to within approximation factor 
$d^{k-1}$ lies in $\NP$ (see also a follow up paper by the same authors 
on approximation versions of problems which are computationally 
even harder \cite{GK2}).



The approach of examining commuting Hamiltonians as a precursor 
for the general case was taken before, e.g. in \cite{Aha2},
and by   
Arad \cite{Arad} who considered the commuting case as a base 
for perturbations to derive a partial no-go theorem for quantum $\PCP$.
Notably, this approach is  
the underlying motivation  
in the recent
works of Hastings \cite{Has} and Hastings and Freedman \cite{Has2}, 
which are tightly related to the current paper, 
as they both attempt to make progress on the $\qPCP$ conjecture  
by studying approximations of $\CLH$s. 
Freedman and Hastings are searching for Hamiltonians whose low 
energy states are all highly entangled, and in particular cannot be 
generated by constant depth quantum circuits (such Hamiltonians are called 
$\NLTS$ for No Low-energy Trivial States). The existence of such Hamiltonians 
is essential for the 
$\qPCP$ conjecture to hold, for otherwise it would be possible to provide 
a classical witness for the fact that the 
energy of the Hamiltonian is low, 
placing the problem in $\NP$. Freedman and Hastings \cite{Has2}
construct a family of $\CLH$s which is not $\NLTS$ but only ``one-sided'' 
$\NLTS$ (see also \cite{AAV} for more on this); 
they leave as an open problem whether true $\NLTS$ Hamiltonians exist.  
Theorem \ref{thm:approxBPinNP} can be applied here to derive 
an energy {\it upper bound} above which such a construction 
{\it cannot} be $\NLTS$, since it shows the existence of trivial states 
below a certain threshold constant fraction of energy. 
We note that this upper bound tends to $0$ 
as the local expansion of the graphs underlying the construction improves, 
and so one cannot hope to achieve better $\NLTS$ constructions  
by improving the local expansion of the underlying graph.  

We finally mention an important recent work discovered independently 
and around the same time as ours, by  
Brand{\~a}o and 
Harrow \cite{HB}. 
Brand{\~a}o and 
Harrow showed that 
the approximation of $2$-local Hamiltonians on graphs whose degree 
is very large lies in $\NP$ (the factor of approximation depends 
inverse polynomially on the degree).
This also implies that when the underlying graph of the $2$-local 
Hamiltonian is an excellent expander (which means that its degree must 
be very large) the approximation problem (to within a factor which 
depends inverse polynomially on the degree again) lies in $\NP$. 

Both the result of \cite{HB} and our main Theorem \ref{thm:approxBPinNP} 
have the same 
flavor: too good expanders are an obstruction to quantum 
hardness, contrary to what one might expect given 
what we know 
in the classical case (\cite{Din}, and Theorem \ref{thm:NPhard}). 
The comparison between the two results, however, 
is not straight-forward, and in particular none of the results implies the 
other. First, Theorem \ref{thm:approxBPinNP} holds for any constant
locality $k\ge 2$, while 
the result of \cite{HB} holds only for $k=2$. One might ask what does 
Theorem \ref{thm:approxBPinNP} give if we substitute $k=2$; 
this turns out to be 
non interesting from the approximation point of view:
it is exactly the result 
\cite{Bra}, which states that the {\it exact} $2$-local $\CLH$ problem 
lies in $\NP$.  
On the other hand, it is unclear whether the 
results of \cite{HB} can be extended to 
$k$-local Hamiltonians, for $k>2$. 
The straight forward attempt to do this, using quantum 
gadgets \cite{KKR,TO, Terhalk},
that transform $k$-local Hamiltonians $(k>2)$ into 
$2$-local Hamiltonians 
encounters an immediate obstacle: 
these gadgets (e.g. the best known gadgets so far \cite{Terhalk}) 
significantly change the geometry of the graph, and in particular, 
may modify its expansion properties considerably.
Specifically, they 
introduce large sets of vertices with no edges between them,
with some vertices having very small degree; 
this implies that the results of \cite{HB} will no longer be applicable after 
such transformations. 

\subsection{Discussion and Further directions}\label{sec:discussion}
Our main result implies that surprisingly, good local expansion is an 
obstruction to quantum hardness of approximation, in the case of commuting 
local Hamiltonians. It implies that if $\qPCP$ were to hold for $\CLH$s, 
then the output of the $\qPCP$ reductions cannot be good local expanders. 
Our result also directly implies that 
$\NLTS$ Hamiltonians (see discussion in the previous subsection  
as well as \cite{Has, AAV}) cannot be constructed on 
too good bi-partite local-expanders.   

The main insight underlying this work is that local Hamiltonian systems 
made of commuting check terms, are 
severely limited by a phenomenon called {\it monogamy of entanglement}.
Essentially, the monogamy of entanglement limits the amount of entanglement 
that a qudit with $O(1)$ quantum levels can "handle".
Our study of commuting system thus reveals that monogamy of entanglement is a significant limiting factor, whose influence
increases together with expansion.
It is intriguing that a similar phenomenon has been identified by completely different methods, by \cite{HB} in the context of general (non-commuting) local Hamiltonians of locality $k=2$.

A natural open question is to generalize both our result and that of 
\cite{HB} into one coherent 
stronger result, which would state roughly that  
for general, not necessarily commuting, $k$-local Hamiltonians 
(for general $k$), the  
expansion of the underlying interaction graph 
(where expansion is defined suitably) 
is an obstruction against quantum hardness. 
Another problem, is to generalize the results here and those of \cite{HB} 
to richer forms of entanglement. 
In particular the limitation on quantum hardness need not necessarily come in the form of constant-depth circuit,
but rather any form of entanglement which has an efficient description.

We mention another interesting direction: 
one might hope that Theorem \ref{thm:approxBPinNP}
could be extended to 
show that the approximation of the $\CLH$ problem on a {\it general} graph 
is in $\NP$, by ``bridging'' between 
our results on locally-expanding graphs and the easy case of approximation on lattices. 
We note that a hint about the difficulty of such a result comes 
from the status of the Unique-Games-Conjecture \cite{Khot} which is 
known to be easy on both extreme cases \cite{Ar}
(albeit for
standard expander graphs rather than 
for locally-expanding graphs), 
though conjectured (or rather, speculated) to be 
$\NP$-hard in general. 
Still, it may be possible to achieve some weaker 
approximation for general graphs using ``bridging'' between the case 
of lattices and expanders, perhaps using 
sub-exponential rather than polynomial witnesses, 
following the works of \cite{Tre} and \cite{ABS}.


{~}

\noindent
\textbf{Organization of paper} 
In Section \ref{sec:bg} we provide some required background. 
Section \ref{sec:approxsec}  
proves Theorem \ref{thm:approxBPinNP} as well as Theorem \ref{thm:NPhard}.  

\section{Definitions and Basic Lemmas}\label{sec:bg}

\subsection{Local Hamiltonians}
A $(k,d)$-local Hamiltonian on $n$ qudits
is a Hermitian matrix $H = \sum_i H_i$ operating on a Hilbert space 
${\cal H} = {\mathbf{C}^d}^{\otimes n}$ 
of $n$ $d$-dimensional particles, where we assume for all $i$, 
$\left\|H_i\right\|=1$ 
and $H_i$ can be written as:
$H_i = h_i \otimes I$, where $h_i$ is a Hermitian operator acting 
on at most $k$ particles. 
When we say an operator acts on some particle, we mean it acts non-trivially, 
namely, it cannot be written as an identity on that particle tensor with 
some operator on the rest. 

\begin{definition}
\textbf{The local Hamiltonian {$\LH$} and
commuting local Hamiltonian ($\CLH$) problems}
In the $(k,d)$-local Hamiltonian problem on $n$ $d$-dimensional qudits 
we are given a $(k,d)$-local Hamiltonian $H=\sum_{i=1}^m H_i$
and two 
constants $a,b$, $a-b = \Omega \left(\frac{1}{poly(n)}\right)$. We are asked to 
decide whether the lowest eigenvalue of $H$ is at least $b$ or at most $a$, 
and we are promised that one of the two occurs. 
In the commuting case, we are guaranteed in addition that $[H_i,H_j]=0$ 
for all $i,j$, and W.L.O.G, we also assume that $H_i$'s are projections. 
$a$ and $b$ can be taken to be $0$ and $1$ in this case, respectively.
An instance to this problem is called a $\CLH(k,d)$ instance.  
\end{definition}

\subsection{Interaction graphs and their expansion} 

We define bi-partite expanders, similar to 
\cite{Spi}, \cite{CRVW}, who used them to 
construct locally-testable classical codes. 
Note that we require expansion to hold only for sets of constant size $k$.  

\begin{definition}\label{def:bipgraph}
\textbf{Bi-Partite Interaction Graph}
For a $(k,d)$-local Hamiltonian $H = \left\{H_i\right\}_i$ 
we define the bi-partite interaction graph $G=(L,R;E)$ as follows: 
$L$, the nodes on the left, correspond to the $n$
particles of ${\cal H}={\mathbf{C}^d}^{\otimes n}$ 
and 
$R$ corresponds to the set of local terms $\left\{H_i \right\}_i$. 
An edge exists between a constraint $r\in R $ and a particle $l\in L$ 
if $H_r$ acts non-trivially (namely, not as the identity) on 
$l$. Note that the right degree is at most $k$. 
\end{definition}

\begin{definition}\label{def:expbi}
\textbf{Bi-partite local expansion}
\noindent
Given a bi-partite graph $G(L,R:E)$, we say that a subset $S\subseteq L$
is $\epsilon$-expanding if $|\Gamma(S)|\ge (1-\epsilon) \sum_{v\in S} D_v$, 
where $D_v$ is the degree of $v$, 
and $\Gamma(S)$ is the set of neighbors 
(in $R$) of $S$.    
A bi-partite graph $G=(L,R;E)$ for a $(k,d)$-local $\LH$ 
is said to be $\epsilon$-locally-expanding, if every subset of particles 
$S\subseteq L$ of size $|S|\leq k$ is $\epsilon$-expanding. 
\end{definition}

\subsection{Commuting terms and \texorpdfstring{$C^*$}{TEXT}-algebras}\label{sec:bv}
We recall the main
lemma of \cite{Bra}: 

\begin{lemma}\label{lem:bvsum}
{\bf Adapted from \cite{Bra}} 
Consider 
two commuting local terms $H_i,H_j$, which 
intersect on some subset of qudits, whose Hilbert space is 
${\cal H}_{int}$. Then 
${\cal H}_{int}$ can be written as a direct sum of subspaces 
\begin{equation}
{\cal H}_{int} = \bigoplus_{\alpha}{\cal H}_{int}^\alpha=
\bigoplus_{\alpha} {\cal H}_{int}^{\alpha,i}\otimes {\cal H}_{int}^{\alpha,j}
\end{equation}
where $H_i$ and $H_j$ preserve the 
${\cal H}_{int}^\alpha$ subspaces, 
and moreover, when $H_i,H_j$ are restricted to 
${\cal H}_{int}^\alpha$, then $H_i (H_j)$ acts non trivially 
only on subsystem ${\cal H}_{int}^{\alpha,i}$ (${\cal H}_{int}^{\alpha,j}$). 
In particular, in each summand $\alpha$, both
${\cal H}_{int}^{\alpha,i},{\cal H}_{int}^{\alpha,j}$ have dimension
strictly less than ${\cal H}_{int}$.
\end{lemma}
\noindent
Thus, under an isometry on ${\cal H}_{int}$ and a restriction to 
one of the ${\cal H}_{int}^\alpha$'s, the terms act on disjoint systems. 
See Figure \ref{fig:ex2}.  

\begin{figure}[ht]
\center{
 \epsfxsize=2.5in
 \epsfbox{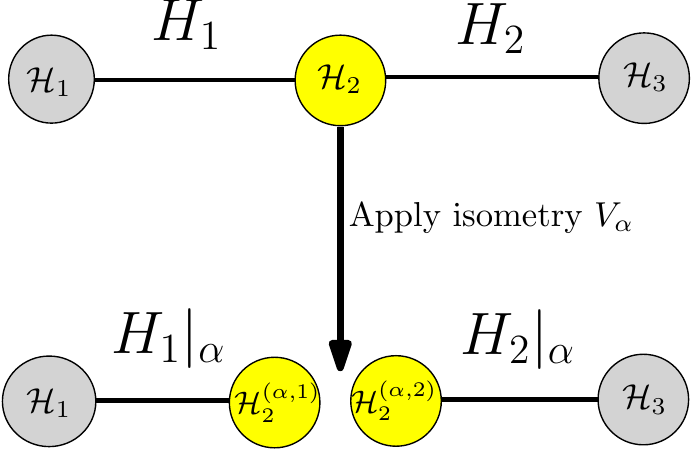}}
 \caption{\label{fig:ex2} 
An example of lemma (\ref{lem:bvsum}): a pair of $2$-local Hamiltonians, are restricted to a subspace $\alpha$ on their intersection.  Inside
this subspace, they act on separate subsystems.} 
\end{figure}
\noindent

\noindent
In this paper we apply 
a simple corollary of Lemma (\ref{lem:bvsum}):  
\begin{corollary}\label{cor:BV}
Let $H$ be a $\CLH(k,d)$ instance, and let $H_0$ be a 
local term such that for any 
$H_j\in H \backslash \left\{H_0\right\}$ that intersects
$H_0$, the intersection is at most on one qudit.
Then for each qudit $q$  participating in $H_0$ there exists a 
direct-sum decomposition of its Hilbert space ${\cal H}_q = \bigoplus_{\alpha_q}{\cal H}_q^{\alpha_q}$, such that each of these subspaces is
preserved by all local terms of $H$, and 
such that the restriction of all terms in $H$ 
to any subspace ${\cal H}_q^{\alpha_q}$ results in $H_0, H\backslash H_0$ acting on a disjoint subsystems, respectively, 
${\cal H}_q^{\alpha_q,0} \otimes {\cal H}_q^{\alpha_q,1}$.
In particular, for any $\alpha_q$, we have that 
$dim({\cal H}_q^{\alpha_q,0}),dim({\cal H}_q^{\alpha_q,1})$
are strictly less than $dim({\cal H}_q)$.
\end{corollary}

\begin{proof}
For each particle $q$ examined by $H_0$, 
let the term $H_1$ be the sum of all local terms acting on $q$ 
except $H_0$. The lemma can now we applied on $q$. 
Since each term other than $H_0$ intersects $H_0$ at only one location, 
we can apply the lemma simultaneously for all qudits in $H_0$. 
The result follows. 
\end{proof}

Note, that the above corollary does not require that 
any two local terms acting on a qudit in $H_0$, 
intersect only on that qudit. 

\subsection{Notation}
$d$ will denote the dimensionality of particles, capital $D$'s will denote  
degrees in a graph, in particular $D_v$ is the degree of a node $v$. 
Given $S\subseteq R(L)$ in a bi-partite graph,   
$\Gamma(S)$ denotes the neighbor set of $S$ in $L(R)$. 
$\epsilon$ 
will be used to denote the expansion error for bi-partite 
graphs (as in Definition \ref{def:expbi}). 
$\gamma$ will be used to denote the promise gap, or alternatively 
the approximation error we are allowed;  
often this means the fraction of terms we are allowed to throw away
given a $\CLH$ instance. 

\section{Approximate $\CLH$}\label{sec:approxsec}

\subsection{Geometric Preliminaries} 
Our proof makes use of a very simple property of locally-expanding graphs, 
namely that most of the neighboring Hamiltonian 
terms of a set of qudits of size $k$ 
only touch this set at one point.  
We first state a simple fact quantifying this statement  
(similar observations were used also in \cite{CRVW})
and then explain how it is used in the proof. 

\begin{fact} \label{fact:essence}
Consider $S\subseteq L$  in a bi-partite graph $G(L,R:E)$ 
and let $S$ be $\epsilon$-expanding, for $\eps<\frac{1}{2}$.
Then a fraction at most $2\eps$ of all
vertices of $\Gamma(S)$ have degree strictly larger than $1$ in $S$.
\end{fact} 

\begin{proof}
Put $D(S) = \sum_{v\in S} D_v$.
The average degree of a vertex in $\Gamma(S)$ w.r.t. $|S|$
is at most $\frac{D(S)}  {D(S) (1-\eps)} = \frac{1}{1-\eps}$.
Let $\alpha_1$ denote the fraction $|\Gamma_1(S)| / |\Gamma(S)|$, 
where $\Gamma_1(S)$ is the set of neighbors of $S$ with degree exactly 
$1$ with respect to $S$. 
Then 
$$\frac{1}{1-\eps} \ge \alpha_1 1 + (1-\alpha_1) m,$$
where $m$ is the average degree of a vertex with at 
least two neighbors in $S$.
Then by simple algebra

$$\alpha_1(m) \ge 
1 - \frac{1}{m-1}\cdot\frac{\eps}{1-\eps},   
$$
so $\alpha_1(m)$ is a monotonously increasing function of $m$,
and since $m\geq 2$, then $\alpha_1$ is minimized for $m=2$.
Hence, 
$$ \alpha_1 \ge  1-\frac{\eps}{1-\eps}. $$
and since $\eps<1/2$ we have:
$$ \alpha_1 \geq 1 - \eps(1+2\eps) \geq 1-2\eps.$$
\end{proof}

\subsubsection{Overview of the proof of Theorem \ref{thm:approxBPinNP}}
The idea of the proof is that if a constraint $H_i$ satisfies a simple 
condition, namely, that it shares at most  
one particle with any other constraint, then Corollary (\ref{cor:BV}) 
can be applied (with $H_i$ playing the role of $H_0$) to ``disentangle'' 
the action of $H_i$ from the rest of the terms   
(from there we can proceed to other terms iteratively).  
The problem is that $H_i$ does not necessarily satisfy this condition. 
The idea is that because the graph is a good locally-expanding graph, we 
can remove only a small number of terms to achieve the condition.  
More formally, we define:  

\begin{definition}\label{def:isolate}
\textbf{Isolated constraints and particles}
Let $(L,R;E)$ be a bi-partite  graph.
A constraint $g\in R$ is called
\textbf{isolated}
if for any constraint $v\in R$ except $g$, 
we have $|\Gamma(v) \cap \Gamma(g)| \leq 1$.
We define the \textbf{isolation penalty} of
$g$, to be
the minimal number of constraints in $R$ that we need to remove 
so that $g$ is isolated. 
\end{definition}

We thus can ``isolate'' $H_i$ by removing all terms that share 
at least $2$ vertices with $H_i$.  
This is where we use the local-expansion property:
by Fact \ref{fact:essence}
the set of constraints we need to remove 
constitutes just a small portion of 
the terms intersecting $H_i$.  
The remaining terms that intersect $H_i$, 
intersect it at only one particle. Corollary (\ref{cor:BV}) 
can now be applied almost directly 
on $H_i$.   
The final result is that the term $H_i$ can be separated 
from the remaining terms and can be viewed (up to the relevant 
isometries, and after a restriction to one of the subspaces in each qudit) 
as acting on a separated subsystem. 

To approximate the $\CLH$ problem in 
$\NP$, it is not enough to remove for each term, all terms 
that make it isolated; this will require removing way too many 
terms. Instead, we use the above idea iteratively, and count how many 
terms we need to remove more carefully.
At each iteration $t$ we choose a local term $v_t$ and ``isolate'' it 
(see figure \ref{fig:ex1}).
Corollary (\ref{cor:BV})
can then be applied, 
and $v_t$ can be separated from the rest, after restricting 
each of its qudits to its relevant subspaces. 
Some qudits may interact, following this process, with only a single 
local term of the Hamiltonian, thereby
allowing us to remove them.
We can now iterate this process, picking another local term;  
this way we gradually "tear away" local subspaces of particles. 

\begin{figure}[ht]
\center{
 \epsfxsize=6in
 \epsfbox{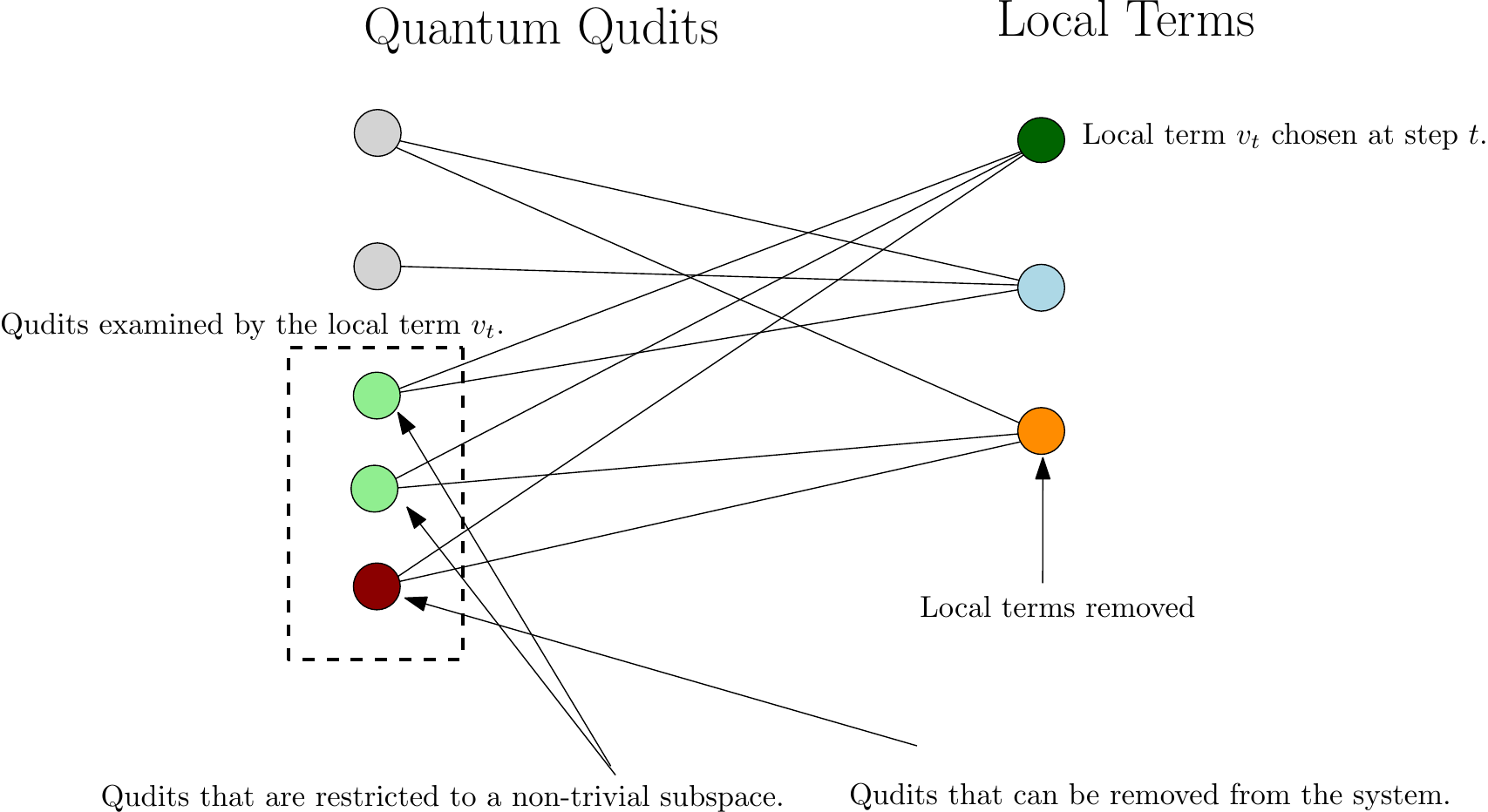}}
 \caption{\label{fig:ex1} 
The isolation procedure w.r.t. $v_t$: we remove the bottom local term since it intersects the neighbor-set of $v_t$ at two locations.
Then we apply lemma (\ref{lem:bvsum}), which removes all degrees of freedom of the bottom particle, while
reducing the dimension of the next two particles by at least $1$.}
\end{figure}

The analysis of the upper bound on the number of terms 
that need to be removed altogether is non-trivial, since 
the number of particles does not necessarily decrease through 
one iteration; it is only their {\it dimensionality} that decreases. 
We resort to an {\it amortized} analysis 
counting the total number of {\it local dimensions} 
(namely, the sum over the particles of the local dimension 
of each particle) that are removed in total. 
To the best of our knowledge this amortized counting of local dimensions 
is novel. 

Note that removing terms here is done not 
in order to {\it disconnect} parts of a graph, but 
only to isolate terms; loosely speaking, to make the interaction pattern
more sparse. This allows a much more economic removal of terms, 
and enables us to achieve the desired result.  

\subsection{Approximating $\CLH$ on local-expanders - 
Proof}\label{sec:appBP}

We are now ready to prove our main theorem:
\begin{proof}(Of Theorem \ref{thm:approxBPinNP})

\noindent
Let $H$ be an instance of $\CLH(k,d)$, and let $(L,R;E)$ be its bi-partite interaction graph. 
We shall find a state whose energy approximates the ground energy of $H$, 
(we note that the same procedure can be applied to the approximation of 
any eigenvalue of $H$).
We perform an iterative process. 
At step $t$, we have a set of remaining terms $R_{rem}(t)$ 
acting on the remaining Hilbert space ${\cal H}_{rem}(t)$. 
We also have a set $R_{bad}(t)$ which are terms we have collected 
so far that we want to throw away. 
We initialize $R_{rem}(1)=R$, i.e. the set of all local terms of $H$, and $R_{bad}(1)$ to be the empty set.
We also initialize ${\cal H}_{rem}(1) = {\cal H}$.
Repeat the following: 
\begin{enumerate} 
\item \textbf{isolate}
Pick an arbitrary local term $v$ in $R_{rem}(t)$.   
and isolate it (definition \ref{def:isolate})
by removing as few as possible terms from $R_{rem}(t)$,
placing them in $R_{bad}(t)$. 
\item \textbf{isometries}
Now that $v$ is isolated, the conditions of Corollary \ref{cor:BV} hold:
there is a tensor product of $1$-local unitaries on 
the qudits of $v$: $W_v = \bigotimes_i W_i$,
such that applying $W_v$ allows to restrict 
the Hilbert space to some 
strict 
subspace containing the zero eigenspace of $R_{rem}(t)$.
Moreover, inside this subspace, the term $v$ acts on a disjoint subsystem
than the rest of $R_{rem}(t)$.
We conjugate $v$, and each term of $R_{rem}(t)$ by $W_v$.
We then {\it prune}: remove any 
triviality we encounter, be it removing qudits from Hamiltonian terms  
or entire terms altogether.
After pruning, each local term acts non-trivially on all the qudits that are attached to it.
\item \textbf{update}
We set $R_{rem}(t+1)$ as the set of remaining 
local terms of $R_{rem}(t)$ after this pruning, 
and then update ${\cal H}_{rem}(t+1)$ as the support of $R_{rem}(t+1)$.
We then set $R_{good}(t+1)$ as the union of $R_{good}(t)$ and the term $v(t)$ restricted to the appropriate subspaces
of its qudits, as in corollary \ref{cor:BV}.  
\end{enumerate} 

We terminate when there is no longer any intersection between two different 
terms of $R_{rem}(t)$.
Clearly, this ends after at most polynomially many 
iterations, since the number of terms in $R_{rem}(t)$ decreases by at least $1$ 
each iteration. 
Let the number of iterations be $T$. 
We claim that a state whose energy is within $|R_{bad}(T)|$  
from the ground energy of the original Hamiltonian,  
can be recovered from this 
procedure by finding the ground state of all terms in $R_{good}(T)$, 
and applying the inverse of the isometries applied along the way. 
This is true as long as the subspaces in the direct sum that are chosen at step $t$, contain 
the groundspace of the current Hamiltonian $R_{rem}(t)$; 
the $\NP$ prover can provide the indices of the subspaces so that 
this holds at each step,
and the verification procedure is simply to check that 
those subspaces contain a non-zero kernel.  

To see that this implies 
a constant-depth quantum circuit that generates a groundstate, 
observe that all isometries on a given qudit $q$, produced
along the way as a result of isolating all local terms from $R_{good}(T)\cup R_{rem}(T)$ on $q$,
commute by definition, and so we can apply them 
in any order, and in particular, simultaneously.
Their application causes the interaction graph to break into linearly many 
connected components of size at most $k$ each, and so each
component can be trivially diagonalized.
The composition of these local diagonalizing unitaries, and the
tensor-product isometries results in a circuit of depth $2$ 
(just as in
\cite{Bra}).


\paragraph{Bounding the approximation error}
We now provide an 
upper bound on the size of the error $|R_{bad}(T)|$. 
For a given $q\in L$, let $D_q$ denote its degree in $G(1)$.
Consider the $t$-th "isolation" step, and denote $v_t$ as the constraint isolated at step $t$,
$S_t = \Gamma(v_t)$,
and $G(t)$ as the bi-partite interaction graph of $L_{rem}(t)$.
We define
for each $q\in S_t$, the quantity $p(t,q)$ which is the number of check terms
removed to $R_{bad}$ during the $t$-th step,  multiplied by the ratio
$\frac{D_q}{\sum_{q\in S_t} D_q}$, where those degrees are taken from $G(1)$.
This accounts for the "relative penalty" shared by $q$, w.r.t. $G(1)$ - 
note that for any iteration $t$, $\sum_{q\in S_t} p(t,q)$ 
is exactly the number of check terms removed at the $t$-th iteration.
Hence 
\begin{equation}\label{eq:totalpenalty}
|R_{bad}(T)|=\sum_{t=1}^T \sum_{q\in S_t} p(t,q). 
\end{equation}
We claim:
\begin{fact}\label{fact:penalty}
For any $1\le t\le T$, and any $q\in S_t$, we have 
$p(t,q)\leq 2\eps D_q$.
\end{fact}

\begin{proof}
We know that in $G(1)$ we had
$|\Gamma(S_t)|\ge (1-\epsilon) \sum_{q\in S_t} D_q$, 
and so by Fact \ref{fact:essence}
in order to isolate $S_t$, it suffices to remove
$2\epsilon \sum_{q\in S_t} D_q$ terms. 
Since $G(t)$ is derived from $G(1)$ by removing vertices and edges, 
this bound still holds -- the isolation penalty cannot increase 
by removing elements from the graph.  
So, at each step, when isolating a term acting on a set of 
particles $S_t$, we remove at 
most  $2\epsilon \sum_{q\in S_t} D_q$ constraints. 
By definition, $p(t,q) \leq 2 \eps D_q$.
\end{proof}

On the other hand, let us calculate how many qudit dimensions we remove 
at such an iteration, out of the total number of qudit dimensions (namely, 
the sum of dimensions over the particles).\footnote{
We note that when we reduce the local dimension of a $d$-level qudit by $1$, 
the dimension of the total Hilbert-space is reduced by a factor of $(d-1)/d$; 
here we are interested however not in the standard dimension of the entire 
Hilbert space but in the {\it sum} of local dimensions over all particles.}  
For each qudit $q\in S_t$, $q$ is either
removed altogether, and thus its local dimension 
is reduced by $dim({\cal H}_q)$,
or it is restricted to some non-trivial
subspace 
and maybe further divided to two sub-particles, one of which is removed
from ${\cal H}_{rem}(t)$; 
in any of these cases, the local
dimension decreases by at least $1$. 
Thus, the total number of local dimensions removed by one application
of "isolate'' is at least $|S_t|$.

We observe by the above
that every qudit $q\in L$ participates in at most $d$ iterations.
We would like to sum its relative penalties $p(t,q)$  
over those iterations.  
By Fact (\ref{fact:penalty}) 
this sum is at most $d \cdot max_t p(t,q) \leq d \cdot 2 \eps D_q$.
We would now like to sum over all $q\in L$; by Equation \ref{eq:totalpenalty}
this will upper bound $|R_{bad}(T)|$. 
We have 
\begin{equation}
|R_{bad}|
\leq
\sum_{q\in L} 2 \eps d D_q = 2\eps d \sum_{q\in L} D_q \leq 2 \eps d k|R|.
\end{equation}
Therefore, the relative penalty satisfies 
$$
\frac{|R_{bad}(T)|}{|R|} \leq 2kd \eps,
$$
as required.
\end{proof} 

\subsection{NP-hardness of approximation}
\label{sec:NPhard}
The proof of Theorem \ref{thm:NPhard} 
relies on the fact that the output 
of the $\PCP$ reductions in Dinur's proof \cite{Din} 
can be made into an excellent locally-expanding graph 
without harming the promise gap, by a small modification of 
her construction.  
We provide here only a sketch of the proof, assuming the reader is 
familiar with \cite{Din}.  Providing a self contained proof
will require essentially repeating Dinur's $\PCP$ proof. 

\begin{proof}(Of Theorem \ref{thm:NPhard})
Let us denote by $\CSP(3,2)$   
$3$-local $\CSP$s acting on bits. Note that instances of 
$\CSP(3,2)$ are special cases 
of $\CLH(k,d)$ for $k\ge 3$ and $d\ge 2$. 
Suppose that for given parameters $k,d$ we are also given $\epsilon$, for which
$\gamma(\epsilon,k,d)\le C_0$. 
The theorem will follow for these parameters $k,d,\eps$ if we show that
it is even $\NP$-hard 
to approximate $\CSP(3,2)$ 
whose bipartite interaction graphs 
are $\epsilon$-locally-expanding , to within $C_0$. 


To do this, recall that 
Dinur provides in \cite{Din} a $\PCP$ reduction which maps 
any $2$-local $\CSP$ instance on $n$ variables of alphabet $\Sigma_0$,
$|\Sigma_0| = O(1)$, 
to a $\CSP$ instance of $poly(n)$ size, 
such that the new instance has constant promise gap, say, $1/10$.  
We modify her construction so that the output $\CSP$ is $3$-local, acting on 
bits, has local expansion error $\epsilon$, and 
the promise gap is larger than $C_0$ where $C_0$ is $1/10$ divided 
by the loss in the promise gap introduced by the alphabet reduction step 
in Dinur's construction; this factor is bounded by a universal constant 
independent of the parameters of the problem \cite{Din}.  

To achieve this we make a small modification of 
Dinur's gap amplification procedure. 
Recall that this procedure consists of iterating three steps: 
the {\it preprocessing} step in which the graph is made into a
$d$-regular expander graph,
{\it gap amplification} in which the $\CSP$ is translated into
another $\CSP$, whose new $2$-local 
constraints are defined using conjunctions of 
constraints in the original $\CSP$ along collections of $t$-walks 
(walks of length $t$), 
thereby significantly amplifying the promise gap of the instance, but
increasing the size of the alphabet, $\ell$, to be doubly exponential in $t$; 
and finally {\it  the alphabet reduction} in which a $2$-CSP
instance alphabet $\Sigma$, is translated
into a $3$-local on bits, with 
some small constant factor 
loss in the promise gap. 
We denote $1/10$ divided by this loss to be $C_0$.    
As a final sub-step, the output $3$-local $\CSP$ is translated
into a $2$-local $\CSP$ on alphabet of size $8$. 
By every such $3$-step iteration, the gap is amplified by some constant larger 
than $1$. 

We apply Dinur's procedure, until the output is a $(3,2)-\CSP$ 
with promise gap at least $1/10$. 
Now we add one additional iteration, in which we 
make two modifications. 
First, we choose the parameter $t$, defined to be
the length of walks in the gap amplification step, in such a way that 
$\ell$, the size of the new alphabet, which is doubly exponential in $t$, 
satisfies $1/L <x\epsilon$ where $L = 2^{2^{2\ell}}$, and $x$ 
is some combinatorial constant to be determined below. 
Second, we skip the very final step where
$3$-local constraints are translated into $2$-local
constraints. 
Finally, regardless of how the alphabet reduction is defined 
in previous iterations (there are several versions of Dinur's proof) apply 
at the final iteration alphabet reduction using Dinur's variant \cite{Din}
of the long code test (\cite{Hastad}). 
We claim that the resulting $\CSP$ is a $3$-local $\CSP$ 
whose underlying bi-partite graph 
has $\epsilon$ local expansion, and whose 
promise gap is at least $C_0$. 

The fact that the promise gap is at least $C_0$ follows from the fact 
that we have applied Dinur's iteration to achieve promise gap at least 
$1/10$, and then the alphabet reduction reduces it to at least $C_0$ 
regardless of the alphabet size (this is how we have defined $C_0$). 
Let us now show why the local expansion error is at most 
$\epsilon$.  
Let $S$ be some set of $3$ bits.
The local expansion error of $S$ can be bounded 
from above by $(1-\alpha_1)/2$, where $\alpha_1=|\Gamma_1(S)|/|\Gamma(S)|$ 
as in the proof of Fact (\ref{fact:essence}). 
Hence, it is sufficient to lower bound 
$\alpha_1$ and show it is at least $1-2\epsilon$.  

We consider in more detail the long-code construction used in the 
alphabet reduction. This step associates with each $2$-local constraint 
acting on two large alphabet variables (each is of dimension $\ell$) with 
$L = 2^{2^{2\ell}}$ bits, 
which are supposed to be the long-code encoding 
of the assignments to those two large alphabet variables. 
Those bits are called $LC$ bits; for each such $L$-tuples of bits, 
the long-code constraints (called $LC$ constraints) are 
$3$-local constraints, one for each triplet among those $L$ bits.  
In addition, we add for each $\ell$-size variable $\ell$ 
bits that 
are responsible for the  {\it Consistency} of this variable in 
all constraints it appears in. These $\ell$ 
bits are thus called $C$-type; 
the consistency constraints ($C$-constraints), 
act on one $C$-type bit, and two $LC$-bits from the long-code 
encoding of a given constraint. We note that choosing any two bits of the 
$C$ constraint determines the third bit. 

We first observe, that the only interesting $S$'s are those that 
have at least $2$ $LC$-type
bits in some long-code encoding of the {\it same} constraint.
This is because in all the other cases, either no constraint can 
intersect $S$ in more than one bit (this is the case when 
$S$ contains no $LC$ bits, 
or it contains $3$ $LC$ bits, all from the long-code encoding 
of different constraints )
or the number of constraints that can intersect $S$ in two bits 
is at most $2$ (this is the case if $S$ contains one $LC$ bit only, 
or if $S$ contains $2$ $LC$ bits, 
but from {\it different} long-code encodings). 
Since all other constraints intersect $S$ in exactly one bit, 
$\alpha_1$ is close enough to $1$. 

So we are left with the case that
$S$ has at least $2$ bits in the same long-code encoding. 
The fraction of $LC$ constraints that intersect $S$ in $2$ 
places, out of all possible $LC$ constraints that intersect it, 
is at most $O(1/L)$;  
The number of $C$ constraints that intersect $S$ in $2$ places, 
is at most $3$ since every two bits determine the third in a $C$ 
constraint. Thus, the fraction of $C$ constraints that intersect $S$ 
in two bits, out of all $C$ constraints that intersect $S$, 
is $O(1/L)$. 
Thus, the expansion error in this case is at most 
$O(1/L)$, which, by choosing $x$ correctly in the choice of $t$ 
above, is $< \eps$.
\end{proof}

\section{Acknowledgements}
The authors would like to thank Irit Dinur and Gil Kalai 
for insightful discussions and Prasad Raghavendra, Luca Trevisan  and Umesh Vazirani for
helpful comments.  
The research leading to these results has received funding from the European
Research Council under the European Union's Seventh Framework Programme
(FP7/2007-2013) / ERC grant agreement n° 280157,
and from
ISF Grant no.

\end{document}